\documentclass[journal,onecolumn,11pt]{IEEEtran}

\makeatletter
\let\NAT@parse\undefined
\makeatother

\usepackage[english]{babel}

\usepackage{amsfonts}
\usepackage[final]{graphicx}
\usepackage{psfrag}
\usepackage{epsfig}
\usepackage[numbers,sort&compress]{natbib}
\usepackage{flushend}

\usepackage{array}
\usepackage{mathrsfs}
\usepackage{amssymb}
\usepackage[cmex10]{amsmath}

\def \P {{\mathcal P }}

\begin{document}
\selectlanguage{english}
%\title{\huge{A New Framework for the Performance Analysis of Wireless Communications under Hoyt (Nakagami-$q$) Fading}
\title{{A New Framework for the Performance Analysis of Wireless Communications under Hoyt (Nakagami-$q$) Fading
}}

\author{Juan M. Romero-Jerez, \emph{Senior Member, IEEE}, and F. Javier Lopez-Martinez, \emph{Member, IEEE}
%\thanks{This work
%was supported in part by the Spanish public Project
%TEC2009-13763-C02-01.}
\thanks{This work has been submitted to the IEEE for possible publication. Copyright may be transferred without notice, after which this version may no longer be accessible.}
\thanks{J. M. Romero-Jerez is with
Departmento de Tecnolog\'{i}a Electr\'{o}nica, E.T.S.I.
Telecomunicaci\'{o}n, Universidad de M\'{a}laga, 29071
M\'{a}laga, Spain (e-mail: romero@dte.uma.es).}
\thanks{F. J. Lopez-Martinez is with Departmento de Ingenier\'{i}a de Comunicaciones, E.T.S.I.
Telecomunicaci\'{o}n, Universidad de M\'{a}laga, 29071
M\'{a}laga, Spain. (e-mail: fjlopezm@ic.uma.es). During the elaboration of this paper, he was with
Wireless Systems Lab, Electrical Engineering,
Stanford University, Stanford, CA.}
%\thanks{The work of F. J. Lopez-Martinez was supported by Junta de Andalucia (P11-TIC-7109), Spanish Government (TEC2013-44442-P, COFUND2013-40259), the University of Malaga and the European Union under Marie-Curie COFUND U-mobility program (ref. 246550).}
% <-this % stops a spac
%\thanks{The material in this paper has been accepted for presentation, in part, at the
%2015 IEEE International Symposium on Information Theory and the 2015 European Wireless conference.}
}

\maketitle
%\renewcommand{\baselinestretch}{1.8}
%\normalsize

%\newpage
\begin{abstract}
%\boldmath

We present a novel relationship between the distribution of circular and non-circular complex Gaussian random variables. Specifically, we show that the distribution of the squared norm of a non-circular complex Gaussian random variable, usually referred to as squared Hoyt distribution, can be constructed from a conditional exponential distribution. From this fundamental connection we introduce a new approach, the \textit{Hoyt transform} method, that allows to analyze the performance of a wireless link under Hoyt (Nakagami-$q$) fading in a very simple way. We illustrate that many performance metrics for Hoyt fading can be calculated by leveraging well-known results for Rayleigh fading and only performing a finite-range integral. We use this technique to obtain novel results for some information and communication-theoretic metrics in Hoyt fading channels.
%
%When analyzing the secrecy capacity of wireless Hoyt fading links we show that the outage secrecy capacity is mainly dominated by the distribution of the SNR at Bob; conversely, it becomes independent of the distribution of the SNR at Eve for sufficiently large values of the average SNR at Bob. 
%
%Then, we investigate the outage probability in Hoyt fading channels in the presence of background noise and arbitrarily distributed interference; we observe how the distribution of the interference impacts the outage probability for low and medium values of the SINR, while it is the distribution of the desired link which has a dominant effect for large values of the SINR.
%
%Finally, when studying the channel capacity for different transmission policies, we observe that the capacity loss compared to Rayleigh fading is under 0.15 bps/Hz for $q>0.5$ when using optimum power and rate adaptation, or optimum rate adaptation with constant power. 
\end{abstract}


\maketitle

\renewcommand{\baselinestretch}{1.8}
\normalsize

\section{Introduction}

The characterization of the distribution of a complex Gaussian random variable (RV) is arguably one of the most relevant problems in engineering and statistics. In the  contexts of information and communication theory, the distribution of the norm of the complex Gaussian random variable $Z=X+jY$ (where $X$ and $Y$ are jointly Gaussian) finds application in many problems such as signal detection, noise characterization, or wireless fading channel modeling, just to name a few.

In the literature, the most general case\footnote{Depending on the specific context, different authors have pursued the characterization of the envelope $E=\sqrt{X^2+Y^2}$ or the squared envelope $R=E^2$. However, both distributions are related by a simple transformation and therefore have a similar form. Throughout this paper, we will focus on the distribution of $R$, which model the instantaneous power of a complex Gaussian signal.} of $X$ and $Y$ having different mean and variance was addressed by Beckmann \cite{Beckmann1967}, as a generalization of the previous results obtained by Rice \cite{Rice1944} and Hoyt \cite{Hoyt1947}, and recently revisited in \cite{Aalo2007,Dharmawansa2009}. In this general situation the chief distribution functions (pdf and cdf) of ${R=X^2+Y^2}$ have complicated forms.

In the specific case of $X$ and $Y$ being independent Gaussian RVs with zero mean and arbitrary variance, or equivalently being correlated Gaussian RVs with zero mean and equal variances, the RV $Z$ is said to be a zero mean non-circular (or improper) Gaussian RV \cite{Picinbono1995}. This occurs in many practical scenarios, such as in the detection of non-stationary complex random signals \cite{Schreier2005}, or in the characterization of multipath fading \cite{Nakagami1960}. In this latter situation, the Hoyt \cite{Hoyt1947} or Nakagami-$q$ fading \cite{Nakagami1960} distribution is used to model short-term variations of radio signals resulting from the addition of scattered waves which can be described as a complex Gaussian RV where the in-phase and quadrature components have zero mean and different variances, or equivalently, where the in-phase and quadrature components are correlated. This distribution is commonly used to model signal fading due to strong ionospheric scintillation in satellite communications \cite{Chytil1967} or in general those fading conditions more severe than Rayleigh, and it includes both Rayleigh fading and one-sided Gaussian fading as special cases. Furthermore, it was shown in \cite{Youssef05} that the second order statistics of Hoyt fading best fit measurement data in mobile satellite channels with heavy shadowing. 

One of the main advantages of Rayleigh fading, which is perhaps the most popular model for the random fluctuations of the signal amplitude when transmitted through a wireless link when there is no direct line-of-sight (LOS) between the transmit and receive ends, relies on its comparatively simple analytic manipulation, as the received signal-to-noise ratio (SNR) is exponentially distributed. Conversely, the received SNR in Hoyt fading has a much more complicated form and sophisticated special functions are required to characterize the pdf or cdf \cite{Paris2009} of a squared Hoyt RV.

Both Rayleigh and Hoyt fading have been extensively investigated in the last few decades \cite{AlouiniBook}; however, while the derivation of information and communication-theoretic performance metrics such as channel capacity \cite{Alouini1999} and outage probability (OP) \cite{Annamalai2001} is usually tractable mathematically for the Rayleigh case, it is way more complicated to analyze the very same scenario when Hoyt fading is assumed.

%For example, while the Shannon capacity of adaptive transmission techniques with diversity combining in Rayleigh channels has been known for years ever since the work by Alouini and Goldsmith \cite{Alouini1999}, the equivalent expressions for Hoyt fading channels have the form of complicated infinite series expressions \cite{Cheng2003,Subadar2010,Khatalin2006} even for single-antenna receivers.

%Another clear example of such inconvenience arises when analyzing the OP with co-channel interference and background noise: the solution for a Rayleigh-distributed fading link with arbitrarily distributed interference is expressed directly in terms of the moment-generating function (MGF) of the aggregate interference \cite{Annamalai2001}. Conversely, the analytical characterization of the OP of a Hoyt-distributed fading link with arbitrarily distributed interference and background noise requires the numerical computation of an inverse Laplace transform \cite{Annamalai2001}.

Dozens of papers have been published in the last years with the aim of analyzing very diverse scenarios where Hoyt fading is considered, for the sake of extending already known results for Rayleigh fading to this more general situation \cite{Cheng2003,Subadar2010,Khatalin2006,Paris2009,Paris2010}. %Compared to Rayleigh, Hoyt fading has an additional degree of freedom by allowing that the in-phase and quadrature components of the complex Gaussian RV that models the channel gain have different variances. 
However, despite the relationship that can be inferred between both distributions, existing analyses in the literature for Hoyt fading do not exploit this connection and usually require tedious and complicated derivations. To the best of our knowledge, there is no standard procedure that takes advantage of the relationship between both distributions, and therefore the calculations for Hoyt fading must be done from scratch.

In this paper, we present a novel connection between the distribution of the squared norm of a non-circular complex Gaussian RV and its circular counterpart. In other words, we introduce a useful relationship between the squared Hoyt distribution and the exponential distribution, which greatly simplifies the analysis of the former. By exploiting the fact that the cdf of a squared Hoyt distributed random variable is a weighted Rice $Ie$-function, we demonstrate that the squared Hoyt distribution can be constructed from a conditional exponential distribution.

This connection has important relevance in practice: Since most communication-theoretic metrics are computed with a linear operation over the SNR distribution, we show that performance results for Hoyt fading channels can be readily obtained by leveraging previously known results for Rayleigh fading, and computing a very simple finite-range integral. This general procedure will be denoted as the \textit{Hoyt transform} method. The main takeaway is that \textit{there is no need to redo any calculation} in order to analyze the performance of communication systems in Hoyt fading, if there are available results for the simpler Rayleigh case. Instead, the application of the Hoyt transform yields the desired performance result in a direct way.

In some cases, the Hoyt transform has analytical solution and hence the expressions for Hoyt fading are of similar complexity to those obtained for Rayleigh scenarios. Otherwise, the results for Hoyt fading have the form of a finite-range integral with constant integration limits, over the performance metric of interest for the Rayleigh case. Integrals of this form are very usual in communications, including proper-integral forms for the Gaussian $Q$-function \cite{Lopez2014}, the Marcum $Q$-function \cite{Simon1998} or the Pawula $F$-function \cite{Pawula1982}, or those obtained with the application of the moment generating function (MGF) approach to the calculation of error probability \cite{Simon1998b}. Therefore, the numerical computation of the Hoyt transform introduced in this paper is simpler than other alternatives that require the evaluation of infinite series or inverse Laplace transforms. As an additional advantage, our new approach also permits to obtain upper and lower bounds of different performance metrics in a simple way. 

Using this general procedure, we provide novel analytical results for three scenarios of interest in information and communication theory: 
\begin{itemize}
\item{We analyze the Shannon capacity of adaptive transmission techniques in Hoyt fading channels, thus extending the results given in \cite{Alouini1999}. Thanks to the Hoyt transform method, we can calculate the asymptotic capacity in the low and high-SNR regimes in closed-form. We show that the asymptotic capacity loss per bandwidth unit in the high-SNR regime is up to 1.83 bps/Hz compared to the AWGN case, and up to 1 bps/Hz when compared to the Rayleigh case.}
\item{We investigate the physical layer security of a wireless link in the presence of an eavesdropper, where both the desired and wiretap links are affected by Hoyt fading. Known analytical results are available for different scenarios such as Rayleigh \cite{Barros2006}, Nakagami-$m$ \cite{Sarkar2009}, Rician \cite{Liu2013}, or Two-Wave with Diffuse Power \cite{Wang2014} fading models. However, to the best of our knowledge, there are no results in the literature for the physical layer security in Hoyt fading channels.}
\item{We evaluate evaluate the OP of a Hoyt-faded wireless link affected by arbitrarily distributed co-channel interference and background noise. Specifically, we show that the OP in this general scenario will be given in terms of a Hoyt transform of the MGF of the aggregate interference, admitting a very simple evaluation even for very general fading distributions such as the $\eta$-$\mu$ or $\kappa$-$\mu$ models \cite{Yacoub2007}.}
\end{itemize}
 
The rest of this paper is organized as follows. In Section \ref{secDef}, some preliminary definitions are introduced and the Rice $Ie$-function is reviewed. Then, the main mathematical contributions are presented in Section \ref{secMain}: the connection between the squared Hoyt and the exponential distributions, and its application to define the Hoyt transform method to the performance analysis in Hoyt fading channels. This approach is used in Sections \ref{secCap} to \ref{secOP} to obtain analytical results in the aforementioned scenarios. Numerical results are presented in Section \ref{secNum}, whereas the main conclusions are outlined in Section \ref{secConc}.

%We first introduce some preliminary definitions that will be useful in our derivations. Then, we proceed to present the main analytical results in this paper.

%%%%%%%%%%%%%%%%%%%%%%%%%%%%%%%%%%%%%%%%%%%%%%%%%%%%%%%%%%%%%%%%%%%%%
%%%% 										SECTION II
\section{Definitions and preliminary results}
\label{secDef}
\subsection{Non-circular complex Gaussian RV}

Let $Z=X+jY$ be a zero-mean non-circular complex Gaussian RV, where $X$ and $Y$ are independent jointly Gaussian RVs with variances $\sigma_x^2$ and $\sigma_y^2$. Then, the random variable ${R=X^2+Y^2}$ is said to follow the squared Hoyt distribution, and its pdf is given by
\begin{equation} \label{eq:01a}
f_R(x)= \frac{1+q^2}{2q \overline{\gamma}} \exp \left[-\frac{(1+q^2)^2 x}{4q^2 \overline{\gamma}}\right]  I_0 \left(\frac{(1-q^4) x}{4q^2 \overline{\gamma}}\right),
\end{equation}
where $I_0(\cdot)$ is the modified Bessel function of the first kind and zero order, $\overline\gamma=\mathbb{E}\{R\}$ and $q=\sigma_y/\sigma_x$, assuming without loss of generality that $\sigma_y\leq\sigma_x$. Therefore, we have $q\in[0,1]$.

If $q=1$, then $Z$ is a zero-mean circular complex Gaussian RV and therefore $R$ is exponentially distributed, with pdf
%\begin{equation} \label{eq:11}
$
	 f_R(x)= \frac{1}{\overline{\gamma}}
	 e^{-x/\overline{\gamma}}.
$
%\end{equation}

\subsection{The Rice $Ie$-function}

%\newtheorem {definicion}{Definition}
%\begin{definicion} 
Let $k$ and $x$ be non-negative real numbers with $0 \leq k  \leq 1$, the Rice $Ie$-function is defined as \cite{Rice1944}
\begin{equation} \label{eq:01}
	Ie(k,x)=\int_0^x e^{-t}I_0(kt) dt.
\end{equation}
%where $I_0(\cdot)$ is the modified Bessel function of the first kind and zero order.
%\end{definicion}

The Rice $Ie$-function admits different infinite series representations \cite{Tan1984,Sofotasios2011}, and it is not considered a tabulated function, in the sense that it is not included as a built-in function in standard mathematical software packages such as Matlab or Mathematica. However, after the appropriate change of notation this function can be written in compact form, as \cite{Pawula1995}
\begin{equation} \label{eq:02}
	Ie(k,x)=\frac{1}{\sqrt{1-k^2}} \left[Q(\sqrt{ax},\sqrt{bx})-Q(\sqrt{bx},\sqrt{ax})\right]
\end{equation}
or equivalently,
\begin{equation} \label{eq:03}
	Ie(k,x)=\frac{1}{\sqrt{1-k^2}} \left[2Q(\sqrt{ax},\sqrt{bx})-e^{-x}I_0(kx)-1\right],
\end{equation}
where $a=1+\sqrt{1-k^2}$, $b=1-\sqrt{1-k^2}$ and
\begin{equation} \label{eq:04}
	Q(\alpha,\beta)=\int_\beta^\infty  t e^{-\frac{t^2+\alpha^2}{2}}I_0(\alpha t) dt
\end{equation}
is the first order Marcum $Q$-function. 

Since both the modified Bessel function $I_0$ and the Marcum $Q$-function are tabulated functions, (\ref{eq:02}) and (\ref{eq:03}) can be easily computed. However, subsequent manipulations of these expressions are generally complicated and in many situations it may be preferable to express the Rice $Ie$-function in integral form. Replacing $I_0(\cdot)$ in (\ref{eq:01}) by its integral representation, namely \cite[eq. (8.431.5)]{Gradstein2007}
\begin{equation} \label{eq:05}
	I_0(z)=\frac{1}{\pi} \int_0^\pi e^{z\cos\theta} d\theta,
\end{equation}
after some manipulation we can write \cite{Tan1984}
\begin{equation} \label{eq:06}
	Ie(k,x)=\frac{1}{\sqrt{1-k^2}}- \frac{1}{\pi} \int_0^\pi \frac{e^{-x(1-k\cos\theta)}}{1-k\cos\theta} d\theta,
\end{equation}
which has important advantages with respect to (\ref{eq:01}), as the integration limits do not depend on the arguments of the defined function, and the integrand is given in terms of elementary functions. However, for reasons that will become clear in the next Section, a much more convenient representation of the Rice $Ie$-function for the purpose of this work is the one provided in the following proposition.

\newtheorem {proposition}{Proposition}
\newtheorem {lemma}{Lemma}
\newtheorem {corollary}{Corollary}

\begin{proposition} 
The Rice $Ie$-function can be written in integral form as
\begin{equation} \label{eq:07}
	Ie(k,x)=\frac{1}{\sqrt{1-k^2}}\left[1 - \frac{1}{\pi} \int_0^\pi \exp \left(-x \frac{{1-k^2}}{1-k\cos\theta}\right) d\theta\right].
\end{equation}
\end{proposition}
\begin{proof}
Although not specifically stated in this form, this result follows from the identities provided by Pawula in \cite[eqs. (2a)-(2d)]{Pawula1995} \cite[eqs. (7)-(10)]{Pawula1998}. In particular, from \cite[eqs. (9)-(10)]{Pawula1998} the following identity holds:
\begin{equation} \label{eq:09}
	 W\int_0^\pi \frac{e^{-(U-V\cos\theta)}}{U-V\cos\theta} d\theta =
	 \int_0^\pi \exp \left(- \frac{{W^2}}{U-V\cos\theta}\right) d\theta,
\end{equation}
where $W=\sqrt{U^2-V^2}$. By identifying $U=x$ and $V/U=k$, and with the help of  (\ref{eq:06}), the desired expression is obtained.
\end{proof}

%%%%%%%%%%%%%%%%%%%%%%%%%%%%%%%%%%%%%%%%%%%%%%%%%%%%%%%%%%%%%%%%%%%%%
%%%% 										SECTION III
\section{Main Results}
\label{secMain}
We now introduce the main mathematical contributions of this work, which are given in a set of lemmas and corollaries.
%\subsection{Integral Representation of the Squared Hoyt Distribution}

%In the following lemma, we present an integral form of the pdf given in (\ref{eq:10}). We will show that it permits to analyze many important performance metrics of wireless links under Hoyt fading using previously known results for Rayleigh fading.  

\begin{lemma} \label{teorema 1}
Let $R|\theta$ be an exponentially distributed random variable, conditioned on $\theta$, with pdf
\begin{equation} \label{eq:11_2}
	 f_{R|\theta}(x)= 	 \frac{1}{\gamma(\theta,q)} e^{-x/\gamma(\theta,q)},
\end{equation}
where $\theta$ is a random variable uniformly distributed between 0 and $\pi$, and
\begin{equation} \label{eq:11_1}
	\gamma(\theta,q) \triangleq \overline{\gamma}\left(1-\frac{1-q^2}{1+q^2}\cos\theta\right)=\mathbb{E}\left\{R|\theta\right\}.
\end{equation}
Then, the unconditional random variable $R$, with pdf  
\begin{equation} \label{eq:11_3}
	 f_R(x)= \frac{1}{\pi} \int_0^\pi 
	 \frac{1}{\gamma(\theta,q)}
	  e^{-x/\gamma(\theta,q)}
	 d\theta,
\end{equation}
follows a squared Hoyt distribution with average $\mathbb{E}\left\{R\right\}=\overline{\gamma}$ and parameter $q$, i.e., (\ref{eq:11_3}) is an alternative expression for the pdf given in (\ref{eq:01a}).
The cdf of $R$ will be given by
\begin{equation} \label{eq:14}
	F_R(x)=1 - \frac{1}{\pi} \int_0^\pi e^{-x/\gamma(\theta,q)} d\theta.
\end{equation}
\end{lemma}

\begin{proof}
The cdf of the $R$ can be calculated as
\begin{equation} \label{eq:12}
	 F_R(x)= \int_0^x \frac{1+q^2}{2q \overline{\gamma}}
	 \exp \left[-\frac{(1+q^2)^2 t}{4q^2 \overline{\gamma}}\right] 
	 I_0 \left(\frac{(1-q^4) t}{4q^2 \overline{\gamma}}\right) dt,
\end{equation}
which can be written using the definition of the Rice $Ie$-function in (\ref{eq:01}) as
\begin{equation} \label{eq:13}
	 F_R(x)= \frac{2q}{1+q^2}
	  Ie \left( \frac{1-q^2}{1+q^2} ,
	  \frac{(1+q^2)^2}{4q^2 \overline{\gamma}} x\right).
\end{equation}

Using the alternative definition for the Rice $Ie$-function in (\ref{eq:07}), the cdf of the SNR can be written after some algebraic manipulation as in (\ref{eq:14}). Finally, by taking the derivative of (\ref{eq:14}), the desired pdf in (\ref{eq:11_3}) is obtained.
\end{proof}

By comparing (\ref{eq:11_3}) with the pdf of an exponential distribution, Lemma \ref{teorema 1} states that a squared Hoyt RV can be viewed as a finite-range integral of a exponentially distributed RV with continuously varying averages. Note that the factor that multiplies $\cos\theta$ in (\ref{eq:11_1}) coincides with the squared third eccentricity $\epsilon$ of the ellipse represented by the underlying non-circular complex Gaussian random variable of the Hoyt distribution \cite{Coluccia2013}, i.e. $\epsilon=\frac{1-q^2}{1+q^2}$.

A direct application of Lemma \ref{teorema 1} in a communication-theoretic context follows: Any performance metric in Hoyt fading channels that can be obtained by averaging over the SNR pdf (e.g. outage probability, channel capacity, error probability) can be calculated from existing results for Rayleigh fading, by performing a finite-range integral. In this situation, $\bar\gamma$ represents the average SNR and $R$ (or $\gamma$, indistinctly) denotes the instantaneous SNR. Since most performance metrics of interest for Rayleigh fading in the literature are usually given in closed-form, the proposed approach allows for easily extending the results to Hoyt fading in a very simple manner. This is formally stated in the following lemma, where the Hoyt transform is introduced.

\begin{lemma} \label{l2}
Let $h(\gamma)$ be a performance metric depending on the instantaneous SNR $R$, and let $\overline{h}_R(\overline{\gamma})$ be the metric in Rayleigh fading with average SNR $\overline{\gamma}$ obtained by averaging over an interval of the pdf of the SNR, i.e.,
\begin{equation} \label{eq:15}
	 \overline{h}_R(\overline{\gamma})= \int_a^b 
	 h(x) \frac{1}{\overline{\gamma}}
	  e^{-x/\overline{\gamma}}
	 dx,
\end{equation}
with $0 \leq a < b \leq \infty$.
Then, the performance metric in Hoyt fading channels with average SNR $\overline{\gamma}$, denoted as $\overline{h}_H(\overline{\gamma})$, can be calculated as
\begin{equation} \label{eq:16}
	 \overline{h}_H(\overline{\gamma})= \frac{1}{\pi} \int_0^\pi 
	\overline{h}_R(\gamma(\theta,q))d\theta=\mathcal{H}\left\{\overline{h}_R(\overline\gamma);q\right\}.
\end{equation}
where $\mathcal{H}\left\{\cdot;q\right\}$ is the Hoyt transform operation.
\end{lemma}
\begin{proof}
The metric $\overline{h}_H(\overline{\gamma})$ is obtained as
\begin{equation} \label{eq:16_1}
	 \overline{h}_H(\overline{\gamma})= \int_a^b 
	 h(x) f(x)
	 dx.
\end{equation}
where $f(x)$ is the pdf of a squared Hoyt random variable given in (\ref{eq:11_3}).
Thus, we can write
\begin{equation} \label{eq:17}
	 \overline{h}_H(\overline{\gamma})= \int_a^b 
	 h(x) \frac{1}{\pi} \int_0^\pi 
	 \frac{1}{\gamma(\theta,q)}
	  e^{-x/\gamma(\theta,q)}
	 d\theta
	 dx,
\end{equation}
and reversing the order of integration\footnote{A sufficient condition for the double integral to be reversible is that $h(x)$ is a nonnegative continuous function, which is the case of most performance metrics of interest, such as channel capacity, symbol error rate, outage probability, etc.
}
yields
\begin{equation} \label{eq:18}
	 \overline{h}_H(\overline{\gamma})= \frac{1}{\pi} \int_0^\pi \left[\int_a^b 
	 h(x)  
	 \frac{1}{\gamma(\theta,q)}
	  e^{-x/\gamma(\theta,q)}
	  dx\right]
    d\theta.
\end{equation}
By recognizing that the integral between brackets is actually $\overline{h}_R(\gamma(\theta,q))$, (\ref{eq:16}) is finally obtained.  
\end{proof}

We see that Lemma \ref{l2} provides a very simple and direct way to analyze the performance of communication systems operating in Hoyt fading channels. In fact, some interesting dualities with the popular MGF approach to the error-rate performance analysis of digital communication systems over fading channels \cite{Simon1998b} can be inferred: 

In the reference work by Simon and Alouini, error-rate expressions are obtained ``\textit{in the form of a single integral with finite limits and an integrand composed of elementary functions, thus readily enabling numerical evaluation}''; in our work, the Hoyt transform also facilitates the derivation of expressions of the same form, with the integrand being now directly the performance metric obtained in the Rayleigh case. However, while the MGF approach is applicable to obtain a specific performance metric (error-rate) in general fading channels; the Hoyt transform approach is applicable to obtain general performance metrics in a specific fading channel (Hoyt).

An interesting consequence of Lemma \ref{l2} is the following corollary.
\begin{corollary}
The MGF of a squared Hoyt random variable of average $\overline{\gamma}$ and shape parameter $q$ can be written as
\begin{equation} \label{eq:00}
	\phi(s)=   
	\frac{1}{\pi} \int_0^\pi 
	\frac {1} {1-\gamma(\theta,q)s} 
	d\theta.
\end{equation}
\end{corollary}
\begin{proof}
This result follows directly from Lemma \ref{l2} and the fact that the MGF of an exponentially  distributed random variable of average $\overline{\gamma}$ is given by $(1-\overline{\gamma}s)^{-1}$.
\end{proof}

This corollary provides an alternative demonstration of the integral representation of the pdf of a squared Hoyt random variable given in (\ref{eq:11_3}). Indeed, because of the way it has been constructed, it is clear that (\ref{eq:00}) is the MGF of a random variable which pdf is given by (\ref{eq:11_3}). On the other hand, the integral in (\ref{eq:00}) can be solved in closed-form, using \cite[eq. (3.613.1)]{Gradstein2007}, yielding
\begin{equation} \label{eq:001}
	\phi(s)=   
	\left[1-2\overline{\gamma}s+\frac{q^2(2\overline{\gamma}s)^2}{(1+q^2)^2} \right]^{-1/2},
\end{equation}
which is the well-known MGF of a squared Hoyt random variable \cite{AlouiniBook}. Therefore, from the uniqueness theorem of the MGF, (\ref{eq:01a}) and (\ref{eq:11_3}) are actually the same pdf.

Another benefit of the Hoyt transform method relies in the fact that the calculations are based on an integration involving a bounded trigonometric function; hence, this permits to find simple upper and lower bounds of the  performance metrics. These bounds can be found by taking into account that symbol error rate performance metrics are usually convex decreasing functions with respect to the SNR, whereas channel capacity metrics are typically concave increasing functions. The following proposition establishes a sufficient condition to determine the monotonicity and convexity of some important average performance metric functions.

\begin{proposition}
Let $h(\gamma)$ be a performance metric depending on the instantaneous SNR $\gamma$ and let $\overline{h}_R(\overline{\gamma})$ be defined as in Lemma 1. 
If $h(\gamma)$ is a decreasing convex (increasing concave) function of $\gamma$ in $[0,\infty)$, then $\overline{h}_R(\overline{\gamma})$ is a decreasing convex (increasing concave) function of $\overline{\gamma}$.
\end{proposition}

\begin{proof}
If $h(\gamma)$ is a decreasing convex function then the first and second order derivatives of $h(\gamma)$ verify ${h'(\gamma)\leq0}$,  $h''(\gamma)\geq0$. 
By a simple change of variables in (\ref{eq:15}), considering the interval $[0,\infty)$, we can write
\begin{equation} \label{eq:19a}
	 \overline{h}_R(\overline{\gamma})= \int_0^\infty 
	 h(\overline{\gamma} x)  e^{-x}	 dx,
\end{equation}
and its first and second order derivatives verify
\begin{equation} \label{eq:19b}
	 \overline{h}'_R(\overline{\gamma})= \int_0^\infty 
	 h'(\overline{\gamma} x) x e^{-x}	 dx < 0,
\end{equation}
\begin{equation} \label{eq:19c}
	 \overline{h}''_R(\overline{\gamma})= \int_0^\infty 
	 h''(\overline{\gamma} x) x^2 e^{-x}	 dx > 0.
\end{equation}
Therefore, $\overline{h}_R(\overline{\gamma})$ is a decreasing convex  function of $\overline{\gamma}$.

Analogously, if $h(\gamma)$ is an increasing concave function, then the first and second order derivatives of $h(\gamma)$ verify $h'(\gamma)\geq0$,  $h''(\gamma)\leq0$. Thus, $\overline{h}_R(\overline{\gamma})$ is an increasing concave function.
\end{proof}

Now we present the aforementioned bounds in the next lemmas:

\begin{lemma} \label{L3} 
Let $\overline{h}_R(\overline{\gamma})$ and $\overline{h}_H(\overline{\gamma})$ be functions obtained by averaging a given function $h(\gamma)$ in, respectively, Rayleigh and Hoyt fading channels, where $\overline{\gamma}$ is the average SNR, and let $\overline{h}_R(\overline{\gamma})$ be a decreasing convex function. 
Then, the following inequality holds: 
\begin{equation} \label{eq:100}
	 \overline{h}_R(\overline{\gamma}) \leq 
	 \overline{h}_H(\overline{\gamma}) 
	 \leq\overline{h}_R  \left(\frac{2q^2}{1+q^2}\overline{\gamma}\right),
\end{equation}
\end{lemma} 

\begin{lemma} \label{L4} 
Let $\overline{h}_R(\overline{\gamma})$ and $\overline{h}_H(\overline{\gamma})$ be functions obtained by averaging a given function $h(\gamma)$ in, respectively, Rayleigh and Hoyt fading channels, where $\overline{\gamma}$ is the average SNR, and let $\overline{h}_R(\overline{\gamma})$ be a concave increasing function. 
Then, the following inequality holds:
\begin{equation} \label{eq:100_1}
	 \overline{h}_R  \left(\frac{2q^2}{1+q^2}\overline{\gamma}\right)  \leq
	 \overline{h}_H(\overline{\gamma}) \leq 
	 \overline{h}_R(\overline{\gamma}).
\end{equation}
\end{lemma} 

\begin{proof}
Let us first demonstrate (\ref{eq:100}): As $\overline{h}_R(\overline{\gamma})$ is a decreasing function of $\overline{\gamma}$ and the lowest value of $\gamma(\theta,q)$ is obtained for $\theta=0$, an upper bound of $\overline{h}_H(\overline{\gamma})$ can be found as
\begin{equation} \label{eq:101}
\begin{split}
	 \overline{h}_H(\overline{\gamma})= &
	\frac{1}{\pi} \int_0^\pi \overline{h}_R(\gamma(\theta,q)) d\theta
	\leq \frac{1}{\pi} \int_0^\pi \overline{h}_R(\gamma(0,q)) d\theta
	 \\ & =\overline{h}_R(\gamma(0,q))=\overline{h}_R  \left(\frac{2q^2}{1+q^2}\overline{\gamma}\right).
\end{split}
\end{equation}
A lower bound of $\overline{h}_H(\overline{\gamma})$ can be found from Jensen's inequality and taking into account that $\overline{h}_R(\overline{\gamma})$ is convex:
\begin{equation} \label{eq:102}
\begin{split}
	 \overline{h}_R(\overline{\gamma})= &
	 \overline{h}_R \left(\frac{1}{\pi}\int_0^\pi \gamma(\theta,q) d\theta\right) \leq \\ &
	\frac{1}{\pi} \int_0^\pi \overline{h}_R(\gamma(\theta,q)) d\theta =
	 \overline{h}_H(\overline{\gamma})
\end{split}
\end{equation}
On the other hand, (\ref{eq:100_1}) can be obtained analogously when $\overline{h}_R(\overline{\gamma})$ is a concave increasing function.
\end{proof}

The bounds in Lemmas \ref{L3} and \ref{L4} state that performance in Hoyt fading, for a given average SNR, will be bounded between that of Rayleigh fading with the same average SNR and that of Rayleigh fading when the average SNR is scaled by a factor $2q^2/(1+q^2)$. Note also that the derived bounds are asymptotically exact as $q \rightarrow 1$.

We have introduced a general approach to the analysis of wireless communication systems operating under Hoyt fading. In the following sections, we use this technique to derive novel results for different performance metrics of interest.

%%%%%%%%%%%%%%%%%%%%%%%%%%%%%%%%%%%%%%%%%%%%%%%%%%%%%%%%%%%%%%%%%%%%%
%%%% 										SECTION IV

\section{Channel Capacity}
\label{secCap}

The channel capacity in Rayleigh fading channels was characterized in \cite{Alouini1999} for different transmission policies. Even though closed-form expressions were attained for the Rayleigh case, the channel capacity in Hoyt fading channels is much more complicated to evaluate. In fact, only infinite series expressions of very complicated argument are available in the literature \cite{Cheng2003,Subadar2010}, which do not facilitate the extraction of any insights. Using the Hoyt transform method, we will now show how to use readily available performance results derived for Rayleigh channels to directly obtain the same performance metric in Hoyt fading.

The capacity per bandwidth unit using optimum rate adaptation (ORA) policy with constant transmit power is calculated as
\begin{equation}
\frac{C_{\text{ora}}}{B}=\overline{C}=\int_{0}^{\infty}\log_2\left(1+\gamma\right)f_{\gamma}(\gamma)d\gamma,
\end{equation}
where $\log$ is the natural logarithm. This capacity metric obtained by averaging the Shannon capacity on a flat-fading channel using the pdf of $\gamma$, and has dimensions of bps/Hz. For a communication system operating under Rayleigh fading with average SNR at the receiver side given by $\overline\gamma$, a closed-form expression for this metric was obtained in \cite[eq. 34]{Alouini1999} as
\begin{equation}
\overline C^{\text{Ray}}=\log_2 (e) e^{1/{\overline\gamma}}E_1(1/{\overline\gamma}),
\end{equation}
where $E_1(\cdot)$ is the exponential integral function. Since $\overline C^{\text{Ray}}$ is computed in the form stated in Lemma \ref{l2}, then the we can directly calculate this metric considering a Hoyt fading channel as
\begin{equation}
\label{eqCora}
\overline C^{\text{Hoyt}}=\frac{\log_2 (e)}{\pi} \int_{0}^{\pi} e^{1/{\gamma(\theta,q)}}E_1\left(\frac{1}{\gamma(\theta,q)}\right)d\theta.
\end{equation}
Note that (\ref{eqCora}) is given in terms of a finite integral over a smooth and well-behaved function. Hence, it can be calculated very accurately. A simple lower bound can be found as
\begin{equation} \label{eq:26b}
	 \overline C^{\text{Hoyt}} \geq \frac{1}{\ln 2}
	 e^{ (1+q^2) /(2q^2\overline{\gamma}) }   
	E_1 \left(\frac{(1+q^2)}{2q^2\overline{\gamma}}\right).
\end{equation}

We now provide asymptotic capacity results in the low-SNR and high-SNR regimes. In the first situation, it is known that the asymptotic capacity in Rayleigh fading is given by \cite{diRenzo2010}
\begin{align}
\label{caplow}
	 \overline C_{\overline\gamma\Downarrow}\approx \log_2 e \frac{d\mathcal{M}_{\gamma}(s)}{ds}|_{s=0} =  \log_2 e \overline\gamma,
\end{align}
where $\mathcal{M}_{\gamma}(s)$ represents the MGF of the SNR. Since (\ref{caplow}) is obtained through linear operations over the distribution of the SNR, we can calculate the asymptotic capacity in Hoyt fading channels in the low-SNR regime as
\begin{align}
\label{caplow_b}
	 \overline C_{\overline\gamma\Downarrow}^{\text{Hoyt}}\approx \frac{\log_2 e}{\pi} \int_{0}^{\pi}\overline\gamma(\theta,q)d\theta = \log_2 e \overline\gamma,
\end{align}
which is the same as in the Rayleigh case, but also the same as in the Rician case \cite{Rao2015}.

In the high-SNR regime, the asymptotic capacity can be expressed in the following form \cite{Yilmaz2012}
\begin{equation}
\label{asy01}
\overline C_{\overline\gamma\Uparrow}\approx \log_2 e \cdot \frac{\partial}{\partial n} \mathbb{E}\left[\gamma^n\right]|_{n=0}.
\end{equation}
which is asymptotically exact. After some manipulations, we can equivalently express (\ref{asy01}) as
\begin{align}
\label{eqasy02}\overline C_{\overline\gamma\Uparrow}&\approx \log_2(e)\cdot \log \overline\gamma - \mu, 
\end{align}
where $\mu$ is a constant value independent of the average SNR, but dependent on the specific channel model. The AWGN case yields a value of $\mu=0$, which is usually taken as a reference. The effect of fading causes $\mu>0$ and therefore there is a non-nil capacity loss due to fading for a given value of $\overline\gamma$. In the case of Rayleigh fading, it is a well-known result that ${\mu=\log_2(e)\cdot\gamma_e}$, where $\gamma_e$ is the Euler-Mascheroni constant. Therefore, for a fixed value of $\overline\gamma$, the parameter $\mu$ can be regarded as the capacity loss with respect to the AWGN case, being $\mu_{Rayleigh}\approx0.83\,\text{bps/Hz}$.

There has been a renewed interest in the research community on different fading models that allow for characterizing a more severe fading condition than Rayleigh fading \cite{Frolik2008,Matolak2011}. Different models, such as the Two-Ray \cite{Frolik2008} or Hoyt fading meet this condition. Very recently, the value of $\mu$ for the Two-Ray fading channel was derived \cite{Rao2015}, yielding a capacity loss of $\mu_{\text{Two-ray}}=1\,\text{bps/Hz}$. This shows that the Two-Ray model is indeed more detrimental than Rayleigh fading, as it provokes a larger capacity loss. However, the value of $\mu$ for Hoyt fading is unknown in the literature to the best of our knowledge. Using Lemma \ref{l2} in (\ref{eqasy02}), we have
\begin{align}
\overline C_{\overline\gamma\Uparrow}^{\text{Hoyt}}&\approx \tfrac{\log_2(e)}{\pi}\int_0^{\pi} \log\left\{\overline{\gamma}\left(1-\tfrac{1-q^2}{1+q^2}\cos\theta\right)\right\}d\theta - \mu_{\text{Rayleigh}}.
\end{align}

After some straightforward manipulations, we have
\begin{align}
\overline C_{\overline\gamma\Uparrow}^{\text{Hoyt}}&\approx \log_2(e)\log \overline\gamma - \mu_{\text{Rayleigh}} - \log_2(e)\log\left\{\tfrac{2(1+q^2)}{(1+q)^2}\right\}.
\end{align}
Therefore, the capacity loss in Hoyt fading with respect to the AWGN case is given by
\begin{align}
\label{asy03}
\mu_{\text{Hoyt}} = \log_2(e)\gamma_e + \log_2(e)\log\left\{\frac{2(1+q^2)}{(1+q)^2}\right\}.
\end{align}
The second term in (\ref{asy03}) can be regarded as the capacity loss with respect to the Rayleigh case, and therefore equals 0 if $q=1$. In the limit case of $q=0$, corresponding to the one-sided Gaussian distribution, we have the larger capacity loss given by
\begin{align}
\label{asy04}
\mu_{\text{Hoyt}}^{q=1} = \log_2(e)\gamma_e +1 \approx 1.83 \text{bps/Hz}.
\end{align}
Therefore, the capacity loss of Hoyt fading with respect to AWGN can be as large as $1$ bps/Hz more than in the Rayleigh case. These results are new to the best of our knowledge.

%%%%%%%%%%%%%%%%%%%%%%%%%%%%%%%%%%%%%%%%%%%%%%%%%%%%%%%%%%%%%%%%%%%%%
%%%% 										SECTION V

\section{Outage Probability of Secrecy Capacity}
\label{secSecr}
\subsection{Problem Definition}

We consider the problem in which two legitimate peers, say Alice and Bob, wish to communicate over a wireless link in the presence of an eavesdropper, say Eve, that observes their transmission through a different link. Let us denote as $\gamma_b$ the instantaneous SNR at the receiver for the link between Alice and Bob, and let $\gamma_e$ be the instantaneous SNR at the eavesdropper for the wiretap link between Alice and Eve.

Unlike the classical setup for the Gaussian wiretap channel \cite{Cheong1978}, it is known that fading provides an additional layer of security to the communication between Alice and Bob \cite{Barros2006,Bloch2008}, allowing for a secure transmission even when Eve experiences a better SNR than the legitimate receiver Bob.

According to the information-theoretic formulation in \cite{Barros2006}, the secrecy capacity in this scenario is defined as
\begin{equation}
C_S=\left[C_B-C_E\right]^+,
\label{eq:sec01}
\end{equation}
 where $[a]^+ \equiv \max \{ a,0 \}$, $C_B$ is the capacity of the main channel
\begin{equation}
C_B=\log\left(1+\gamma_b\right),
\label{eq:sec02}
\end{equation}
and $C_E$ is the capacity of the eavesdropper channel
\begin{equation}
C_E=\log\left(1+\gamma_e\right).
\label{eq:sec03}
\end{equation}
For the sake of simplicity, we assumed a normalized bandwidth $B=1$ in the previous capacity definitions. Since
channel capacity is by definition a non-negative metric, the secrecy capacity for a given realization of the fading links is therefore given by
\begin{equation} \label{eq:sec04}
C_S = \left[\log \left( {1 + \gamma _b } \right) - \log \left( {1 + \gamma _e } \right)\right]^+.
\end{equation}

In \cite{Barros2006,Bloch2008}, the physical layer security of the communication between Alice and Bob in the presence of Eve was characterized in terms of several performance metrics of interest, assuming that both wireless links undergo Rayleigh fading. Specifically, closed-form expressions were derived for the probability of strictly positive secrecy capacity ${\P(C_S>0)}$, and for the outage probability of the secrecy capacity ${\P(C_S<R_S)}$, where $R_S$ is defined as the threshold rate under which secure communication cannot be achieved. As these expressions will be used in the forthcoming analysis, we reproduce them for the readers' convenience
\begin{align}
\label{eq:sec05}
\P(C_S>0)&=\frac{\bar{\gamma}_b}{\bar{\gamma}_b+\bar{\gamma}_e},\\
\label{eq:sec06}
\P(C_S<R_S)&=1-\frac{\bar{\gamma}_b}{\bar{\gamma}_b+2^{R_S}\bar{\gamma}_e}\exp{\left(-\frac{2^{R_S}-1}{\bar{\gamma}_b}\right)},
\end{align}
where $\bar{\gamma}_b$ and $\bar{\gamma}_e$ are the average SNRs at Bob and Eve, respectively. 

We note that ${\P(C_S>0)=1-\P(C_S<R_S)_{R_S=0}}$; hence, the probability of strictly positive secrecy capacity will be considered as a particular case of the secrecy outage probability.

\subsection{Secrecy Outage Probability Analysis}

Let us consider the scenario where the wireless links experience a more severe fading than Rayleigh, say Hoyt, where $q_b$ and $q_e$ represent the Hoyt shape parameters for the desired and eavesdropper links, respectively. We also define the eccentricities associated with both Hoyt distributions as $\epsilon_b=\frac{1-q_b^2}{1+q_b^2}$ and $\epsilon_e=\frac{1-q_e^2}{1+q_e^2}$

According to the Hoyt transform, the outage probability of the secrecy capacity in Hoyt fading channels is given by
\begin{align}
\label{eq:sec07}
\P(C_S  < R_S ) = &1 - \frac{1}{{\pi ^2 }}\int_0^\pi  \int_0^\pi  \exp \left( { - \tfrac{{2^{R_S }  - 1}}{{\bar \gamma _b \left( {1 - \epsilon _b \cos \theta _b } \right)}}} \right)  \\& \nonumber \times
\tfrac{{\bar \gamma _b \left( {1 - \epsilon _b \cos \theta _b } \right)}}{{\bar \gamma _b \left( {1 - \epsilon _b \cos \theta _b } \right) + 2^{R_S } \bar \gamma _e \left( {1 - \epsilon _e \cos \theta _e } \right)}}d\theta _e d\theta _b.
\end{align}
We observe that the integral over $\theta_e$ can be solved, and hence we obtain
\begin{align}
\label{eq:sec08}
P(C_S  < R_S ) =& 1 - \frac{1}{\pi }\int_0^\pi \exp \left( { - \tfrac{{2^{R_S }  - 1}}{{\bar \gamma _b \left( \theta  \right)}}} \right) \\& \nonumber \times
\tfrac{{\bar \gamma _b \left( \theta  \right)}}{{\bar \gamma _b \left( \theta  \right) + 2^{R_S } \bar \gamma _e }}\tfrac{1}{{\sqrt {1 - \left( {\frac{{\epsilon _e 2^{R_S } \bar \gamma _e }}{{\bar \gamma _b \left( \theta  \right) + 2^{R_S } \bar \gamma _e }}} \right)^2 } }}d\theta,
\end{align}
where ${\bar \gamma _b \left( \theta  \right)}={\bar \gamma _b \left( {1 - \epsilon _b \cos \theta} \right)}$. Hence, the secrecy capacity OP is given in terms of a very simple integral form. This result is new in the literature to the best of our knowledge, and shows the strength and versatility of the Hoyt transform method to derive new performance metrics for Hoyt fading by leveraging existing results for Rayleigh fading.

Directly from (\ref{eq:sec08}), the probability of strictly positive secrecy capacity can be easily obtained as
\begin{align}
\label{eq:sec09}
P(C_S {\rm{ > }}0) = \frac{1}{\pi }\int_0^\pi  {\tfrac{{\bar \gamma _b \left( \theta  \right)}}{{\bar \gamma _b \left( \theta  \right) + \bar \gamma _e }}\tfrac{1}{{\sqrt {1 - \left( {\frac{{\varepsilon _e \bar \gamma _e }}{{\bar \gamma _b \left( \theta  \right) + \bar \gamma _e }}} \right)^2 } }}} d\theta.
\end{align}

Expressions (\ref{eq:sec08}) and (\ref{eq:sec09}) admit an easy manipulation, in order to extract insights on the effect of fading severity into the secrecy capacity OP. One clear example arises if we assume that the eavesdropper link suffers from a more severe fading compared to the desired link: this can be achieved by setting $q_b=1$, and seeing what is the impact of $q_e$. In this case, the integral over $\theta$ disappears, yielding to a closed-form expression for the secrecy capacity OP, and hence  
\begin{align}
\label{eq:sec09b}
\left. {P(C_S  < R_S )} \right|_{q_b = 1}  = 1 - \frac{{\bar \gamma _b }}{{\bar \gamma _b  + 2^{R_S } \bar \gamma _e }}\frac{{\exp \left( { - \frac{{2^{R_S }  - 1}}{{\bar \gamma _b }}} \right)}}{{\sqrt {1 - \left( {\frac{{\varepsilon _e 2^{R_S } \bar \gamma _e }}{{\bar \gamma _b  + 2^{R_S } \bar \gamma _e }}} \right)^2 } }}
\end{align}

Comparing (\ref{eq:sec06}) and (\ref{eq:sec09b}), we observe that both have similar form, and the effect of the distribution of the fading for the eavesdropper link is captured by a multiplicative term that modulates the result for the Rayleigh case. Since this additional term is always larger than one, it is clear that for a fixed value of $\bar \gamma _b$ and $\bar \gamma _e$, the secrecy capacity OP $P(C_S  < R_S )$ decreases with $q_e$. This illustrates the fact that when Eve suffers from a more severe fading, then the probability of having a secure communication between Alice and Bob grows.

A similar conclusion can be extracted when examining the probability of strictly positive secrecy capacity in this particular scenario:
\begin{align}
\label{eq:sec10}
\left. {P(C_S  > 0 )} \right|_{q_b = 1}  =\frac{{\bar \gamma _b }}{{\bar \gamma _b  + \bar \gamma _e }}\frac{1}{{\sqrt {1 - \left( {\frac{{\epsilon _e \bar \gamma _e }}{{\bar \gamma _b  + \bar \gamma _e }}} \right)^2 } }}.
\end{align}
Again, the effect of considering $q_e<1$ (i.e. a more severe fading than Rayleigh for the eavesdropper link) causes that ${P(C_S  > 0 )}$ grows as $q_e$ is reduced.

%%%%%%%%%%%%%%%%%%%%%%%%%%%%%%%%%%%%%%%%%%%%%%%%%%%%%%%%%%%%%%%%%%%%%
%%%% 										SECTION VI

\section{Outage Probability}
\label{secOP}
\subsection{Outage Probability in Noise-Limited Scenarios}

The outage probability is one of the most important performance metrics in wireless communications, and it is defined as the probability that the received SNR falls below a predefined threshold $\gamma_o$. Thus, the outage probability is given by ${P_{out}(\gamma_o)=F(\gamma_o)}$, where $F(\cdot)$ is the cdf of the received SNR. 
Therefore, the outage probability under Hoyt fading can be written, from (\ref{eq:14}), as
\begin{equation} \label{eq:120}
P_{out}(\gamma_o)=1 - \frac{1}{\pi} \int_0^\pi e^{-\gamma_o/\gamma(\theta,q)} d\theta,
\end{equation}
which can be very efficiently computed, as the integrand varies smoothly for all possible values of parameter $q$. Conversely, the integrand of the integral representation of the outage probability given in \cite{Tavares2010} becomes sharply peaked when $q$ is close to 0. 

It is also interesting to note that, as the outage probability in Rayleigh fading (given by $ 1- e^{ - \gamma_o /\overline{\gamma} }$) is an increasing concave function with respect to $\overline{\gamma}$, the outage probability in Hoyt fading can be bounded as
\begin{equation} \label{eq:25b}
	 1- e^{ - \gamma_o /\overline{\gamma} } \leq
	 P_{out}(\gamma_o) \leq 1-
	 e^{ - (1+q^2) \gamma_o /(2q^2\overline{\gamma}) }.
\end{equation}
 
Alternatively, the outage probability can be expressed in closed-form using (\ref{eq:02}) and (\ref{eq:13}) as
\begin{equation} \label{eq:23}
\begin{split}	
P_{out}(\gamma_o)= & Q\left(a(q)\sqrt{\frac{\gamma_o}{\overline{\gamma}}},b(q)\sqrt{\frac{\gamma_o}{\overline{\gamma}}}\right)
\\ & -Q\left(b(q)\sqrt{\frac{\gamma_o}{\overline{\gamma}}},a(q)\sqrt{\frac{\gamma_o}{\overline{\gamma}}}\right)
\end{split}
\end{equation}
or, equivalently, from (\ref{eq:03}) and (\ref{eq:13}),
\begin{equation} \label{eq:24}
\begin{split}	
P_{out}(\gamma_o)= & 2Q\left(a(q)\sqrt{\frac{\gamma_o}{\overline{\gamma}}},b(q)\sqrt{\frac{\gamma_o}{\overline{\gamma}}}\right)
	\\ & -e^{-d(q)\gamma_o/\overline{\gamma}}I_0\left(c(q){\frac{\gamma_o}{\overline{\gamma}}}\right)-1,
\end{split}
\end{equation}
where
\begin{equation} \label{eq:25} \nonumber
\begin{split}
& a(q)=\frac{1+q } {2q}\sqrt{1+q^2} , \ 
b(q)=\frac{1-q } {2q}\sqrt{1+q^2},  \\ &
c(q)=\frac{1-q^4 }{4q^2}, \ \ \
d(q)=\frac{(1+q^2)^2} {4q^2},
\end{split}
\end{equation}
which were previously derived in \cite{Paris2009} in a slightly different way. Note, however, that (\ref{eq:120}) is easier to compute than these closed-form expressions because the Marcum $Q$-function is actually a two-fold integral, as the Bessel function is an integral itself. Although a finite-range integral can be used to compute the Marcum $Q$-function \cite[eqs. (4.39), (4.42)] {AlouiniBook}, the integrand is sharply peaked for low values of $q$, which complicates the computation. Mathematical software packages often use truncated infinite power series \cite{Shnidman1989} to compute the Marcum $Q$-function.

\subsection{Outage Probability with Co-Channel Interference}
In many practical scenarios, the signal of interest is affected by co-channel interference. In this case the outage probability can be defined as the probability that the signal-to-interference-plus-noise ratio (SINR) falls below a threshold level $\gamma_o$. More formally, let $\mathcal{X}$ denote the SNR of the signal of interest and $\mathcal{Y}$ the total interference-to-noise ratio (INR). The outage probability will thus be defined as
\begin{equation} \label{eq:40}
P_{out}(\gamma_o)=P(\mathcal{X}<\gamma_o(\mathcal{Y}+1))=F_\mathcal{X}(\gamma_o(\mathcal{Y}+1)),  
\end{equation}
where $F_\mathcal{X}(\cdot)$ is the cdf of $\mathcal{X}$. In scenarios where the background noise can be neglected, the signal-to-interference ratio (SIR) is typically considered instead of the SINR, which usually simplifies the analysis.

We derive in this section a simple expression for the outage probability for the case when the signal of interest undergoes Hoyt fading and there is an arbitrary number of interferers. Our analysis will be quite general, as we assume that background noise is not necessarily neglected and each interferer undergoes an arbitrary fading. The main result for this scenario is presented in the following proposition.

\begin{proposition} \label{lema 3}
Let the desired signal undergo Hoyt fading, and the interfering signals experience arbitrarily distributed fading. Then, the outage probability in the presence of co-channel interference and background noise is given by
\begin{equation} \label{eq:44}
	P_{out}(\gamma_o)=  1 - 
	\frac{1}{\pi} \int_0^\pi e^{-\gamma_o/\gamma(\theta,q)}
	\phi_\mathcal{Y} \left(-\frac {\gamma_o} {\gamma(\theta,q)}\right) 
	d\theta.
\end{equation}
\end{proposition} 
where $\phi_\mathcal{Y}$ is the MGF of $\mathcal{Y}$.

\begin{proof}
As $\mathcal{X}$ follows a squared Hoyt distribution, from (\ref{eq:14}) we have that the conditional (on $\mathcal{Y}$) outage probability can be written as
\begin{equation} \label{eq:41}
	\left.P_{out}(\gamma_o)\right|_{\mathcal{Y}=y}=1 - \frac{1}{\pi} \int_0^\pi e^{-\gamma_o(y+1)/\gamma(\theta,q)} d\theta.
\end{equation}
Averaging over $\mathcal{Y}$, the unconditional outage probability will be
\begin{equation} \label{eq:42}
	P_{out}(\gamma_o)=1 - 
	\int_0^\infty
	\frac{1}{\pi} \int_0^\pi e^{-\gamma_o(y+1)/\gamma(\theta,q)} d\theta f_\mathcal{Y}(y) dy,
\end{equation}
where $f_\mathcal{Y}(\cdot)$ denotes the pdf of $\mathcal{Y}$.
By interchanging the order of integration we can rewrite (\ref{eq:42}) as
\begin{equation} \label{eq:43}
\begin{split}
	P_{out}(\gamma_o)= & 1 - 
	\frac{1}{\pi} \int_0^\pi e^{-\gamma_o/\gamma(\theta,q)}
	\\ & \times \left[\int_0^\infty e^{-y \gamma_o/\gamma(\theta,q)}f_\mathcal{Y}(y) dy \right]
	d\theta.
\end{split}
\end{equation}
By noticing that the MGF of a positive random variable $\alpha$ is defined as
$\phi_\alpha(s)=\int_0^\infty e^{st}f_\alpha(t)dt$, with $f_\alpha (\cdot)$ being the pdf of $\alpha$, (\ref{eq:44}) is obtained.
\end {proof}

Proposition \ref{lema 3} allows to analyze the outage probability in Hoyt fading channels with \textit{arbitrarily distributed} co-channel interference in the presence of background noise. Since the MGF is given in closed-form for the most relevant fading distributions, then (\ref{eq:44}) can be easily evaluated as a finite-range integral. The particular case where the background noise can be neglected is presented in the next corollary:

\begin{corollary}
When the total interference power is much higher that the background noise and the latter can be neglected, the outage probability is given by
\begin{equation} \label{eq:45}
	P_{out}(\gamma_o)=  1 - 
	\frac{1}{\pi} \int_0^\pi 
	\phi_\mathcal{Y} \left(-\frac {\gamma_o} {\gamma(\theta,q)}\right) 
	d\theta.
\end{equation}
\end{corollary}

\begin{proof}
When the background noise can be neglected, the outage probability can be defined as
\begin{equation} \label{eq:46}
P_{out}(\gamma_o)=P(\mathcal{X}<\gamma_o\mathcal{Y})=F_\mathcal{X}(\gamma_o\mathcal{Y}),  
\end{equation}
and by following the same steps as in Proposition \ref{lema 3}, (\ref{eq:45}) is obtained.
\end{proof}

When $L$ independent interferers are considered, the MGF of $\mathcal{Y}$ will be
\begin{equation} \label{eq:47}
\phi_\mathcal{Y}(s)=\prod_{1=1}^L \phi_{\mathcal{Y}_i}(s), 
\end{equation}
where $\phi_{\mathcal{Y}_i}(\cdot)$ is the MGF corresponding to the \emph{i}-th interferer. 

In recent years, several general fading models such as $\eta$-$\mu$ or $\kappa$-$\mu$ have been proposed, showing a better fit to experimental measurements than traditional fading models in many different environments \cite{Yacoub2007}. The $\eta$-$\mu$ model includes Hoyt, Nakagami-$m$, Rayleigh and one-sided Gaussian fading as particular cases, whereas the $\kappa$-$\mu$ model includes Rician, Nakagami-$m$ and Rayleigh as particular cases. In spite of their generality, these models have a MGF that can be expressed in simple terms \cite{Ermolova2008}. The MGF of $\eta$-$\mu$ is given by
\begin{equation} \label{eq:48}
\phi(s)= \left(\frac{4\mu^2 h}{ (2(h-H)\mu -s\overline{\gamma})  (2(h+H)\mu -s\overline{\gamma})  }\right)^\mu,
\end{equation}
and the distribution is defined in two different formats. In format 1 we have $H=(\eta^{-1}-\eta)/4$ and $h=(2+\eta^{-1}+\eta)/4$, with $0<\eta<\infty$, while in format 2: $H=\eta/(1-\eta^2)$ and $h=1/(1-\eta^2)$, with $-1<\eta<1$.

For the $\kappa$-$\mu$ fading model we have
\begin{equation} \label{eq:49}
\phi(s)= \left(\frac{\mu(1+\kappa)} {\mu(1+\kappa)-s\overline{\gamma}}  \right)^\mu
\exp\left(\frac{\mu^2\kappa(1+\kappa)} {\mu(1+\kappa)-s\overline{\gamma}}-\mu\kappa  \right),
\end{equation}
Introducing (\ref{eq:48}) and/or (\ref{eq:49}) as factors in (\ref{eq:47}) we can consider many different and general interfering scenarios.

Note that, as we have not imposed any restriction on the interferers statistics, our outage probability expression also includes the case of correlated interferers, as long as the MGF of the total interference power is known. Fortunately, there exist closed-form expressions for the MGF of the addition of correlated signal powers of traditional fading models, such as Rayleigh, Nakagami-$m$ or Rician \cite{Annamalai2001}, and more general fading models such as $\eta$-$\mu$
\cite{Asghari2010}.

%%%%%%%%%%%%%%%%%%%%%%%%%%%%%%%%%%%%%%%%%%%%%%%%%%%%%%%%%%%%%%%%%%%%%
%%%% 										SECTION VII

\section{Numerical results}
\label{secNum}
After describing how the Hoyt transform method can be applied to obtain new results for some information and communication-theoretic performance metrics, now we discuss the main implications that arise in practical scenarios.

\subsection{Channel capacity}
First, we study the Shannon capacity of ORA technique in Hoyt fading channels. In the following set of figures, we evaluate the capacity per bandwidth unit in Hoyt fading channels using (\ref{eqCora}).

\begin{figure}[lt]
\includegraphics[width=\columnwidth]{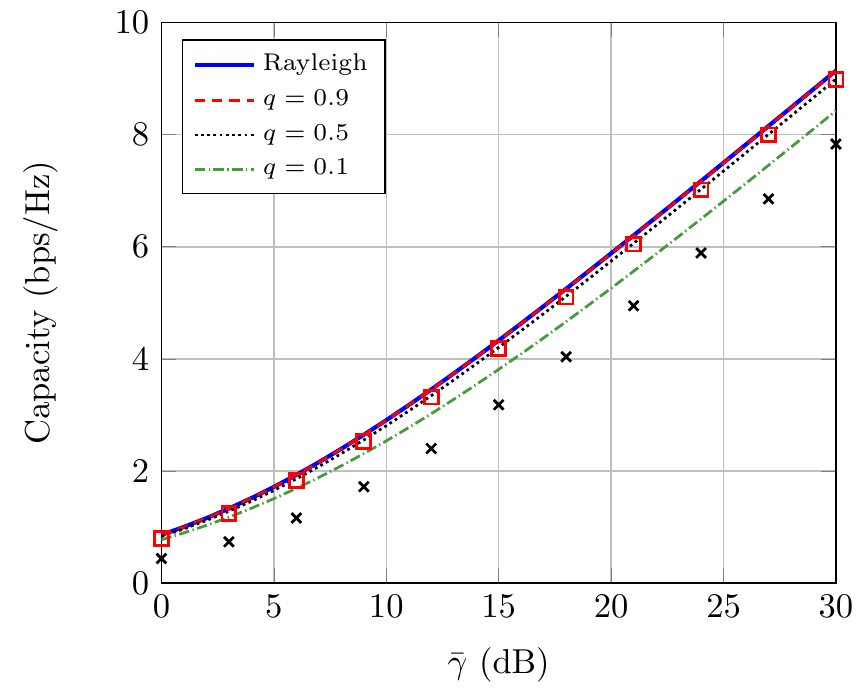}
\caption{Normalized capacity vs $\bar\gamma$ using ORA policy, for different values of $q$. Markers indicate the lower bounds on capacity given by (\ref{eq:26b}), for $q=0.5$ ('x') and $q=0.9$ (squares).}
\label{fig6}
\end{figure}

In Fig. \ref{fig6}, we observe how the capacity loss due to a more severe fading is low for values of $q>0.5$, being under $0.15$ bps/Hz in this range. In fact, it is noted how the achievable performance when $q=0.9$ is practically coincident with the Rayleigh case. Markers indicate the simple lower bounds obtained in (\ref{eq:26b}). We see how the lower bound becomes tighter as $q$ is increased, whereas the performance for the Rayleigh case serves as an upper bound.

The accuracy of the asymptotic approximations yielding from (\ref{asy01}) in the high-SNR regime is evaluated in Fig. \ref{fig1new}. The capacity of the AWGN and Rayleigh ($q=1$) cases are included as reference.
\begin{figure}[lt]
\includegraphics[width=\columnwidth]{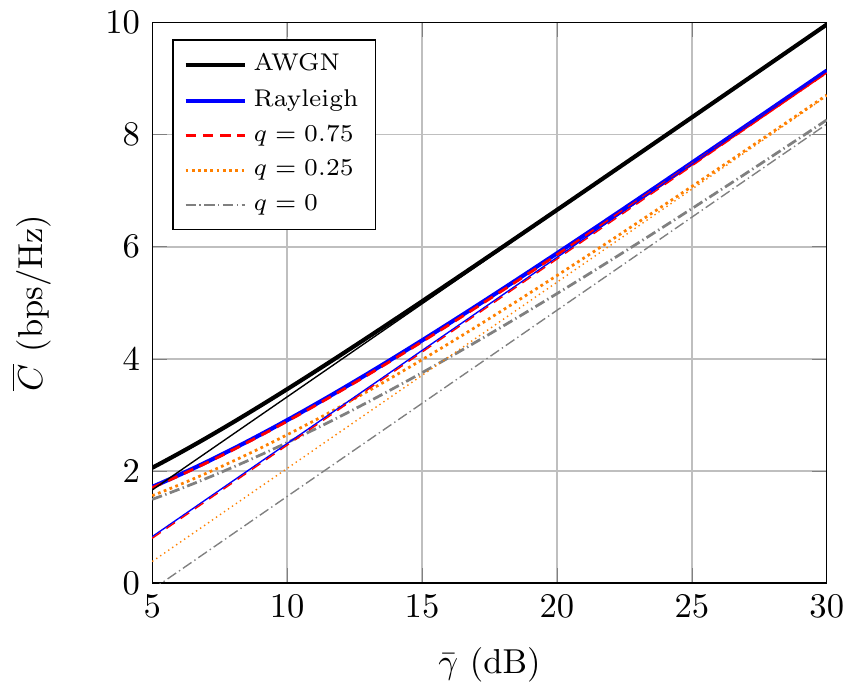}
\caption{Capacity per bandwidth unit vs $\overline\gamma$ using ORA policy, for different values of $q$. Straight thin lines represent the asymptotic results for high-SNR.}
\label{fig1new}
\end{figure}

We observe that for values of $q$ close to 1, the capacity loss with respect to the Rayleigh case is almost negligible. However, we see that as $q$ is reduced, the gap between the capacities in the Rayleigh and Hoyt cases grows. We see how the approximations for the capacity are asymptotically exact for sufficiently large $\bar\gamma$. However, we also observe that for $q=0$ the convergence between the exact and asymptotic capacity takes place at a larger value of $\bar\gamma$. We must note that the asymptotic results obtained for Hoyt fading are also valid for communication systems using optimal power and rate adaptation (OPRA) policy, since in the high-SNR regime this policy has the same performance as ORA \cite{Alouini1999}.

In Fig. \ref{fig2new}, we study the asymptotic capacity loss due to  Hoyt fading with respect to the reference AWGN case, which is given by (\ref{asy03}). The capacity loss for Rayleigh and Two-Ray \cite{Rao2015} are also included for comparison purposes.
\begin{figure}[lt]
\includegraphics[width=\columnwidth]{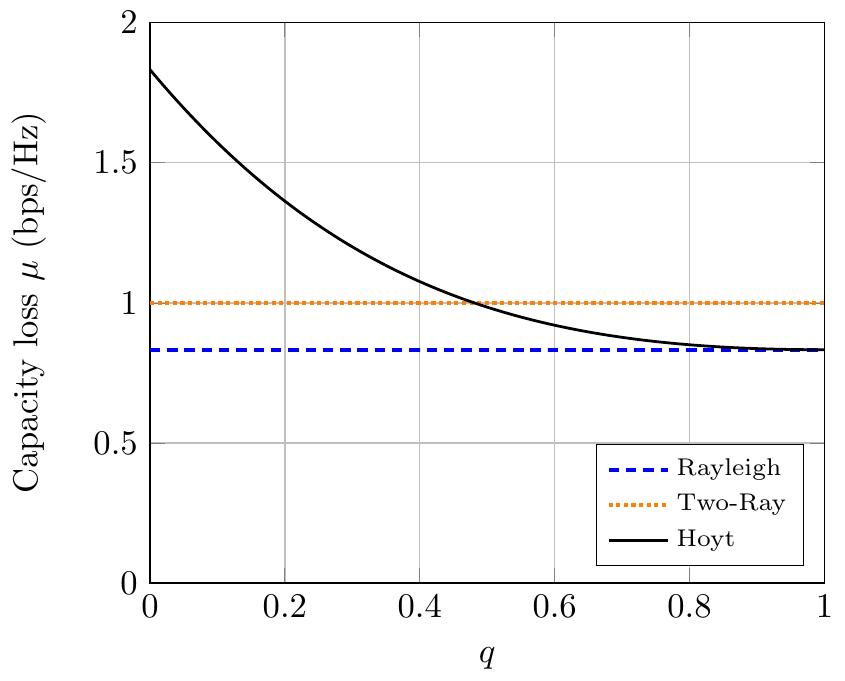}
\caption{Asymptotic capacity loss $\mu_{\text{Hoyt}}$ as a function of $q$.}
\label{fig2new}
\end{figure}
As indicated in (\ref{asy04}), the maximum capacity loss is obtained when $q=0$. This corresponds to the more severe fading condition that Hoyt fading can represent. As $q\rightarrow 1$, we see that the capacity loss tends to the value obtained in \cite{Alouini1999}, i.e. approximately $0.83$ bps/Hz. The comparison with the Two-Ray fading model is also interesting: we see that the capacity loss in Hoyt and Two-Ray fading is coincident for $q\approx0.48$. This suggests that Hoyt fading represents a more severe condition than the Two-Ray model if $q<0.48$, when the asymptotic capacity loss is chosen as the performance metric of interest.

\subsection{Secrecy capacity}

Now we focus on the scenario considered in Section \ref{secSecr}; specifically, we will evaluate the effect of considering that the links between Alice and Bob (and equivalently between Alice and Eve) can suffer from different fading severities, quantified by the parameters $q_b$ and $q_e$. In Fig. \ref{fig1}, the secrecy capacity OP derived in (\ref{eq:sec08}) is represented as a function of the average SNR at Bob $\bar\gamma_b$, for different sets of values of the Hoyt shape parameters. We assume that the normalized rate threshold value used to declare an outage is $R_S=0.1$, and an average SNR at Eve $\bar\gamma_e=15$ dB.

\begin{figure}[lt]
\includegraphics[width=\columnwidth]{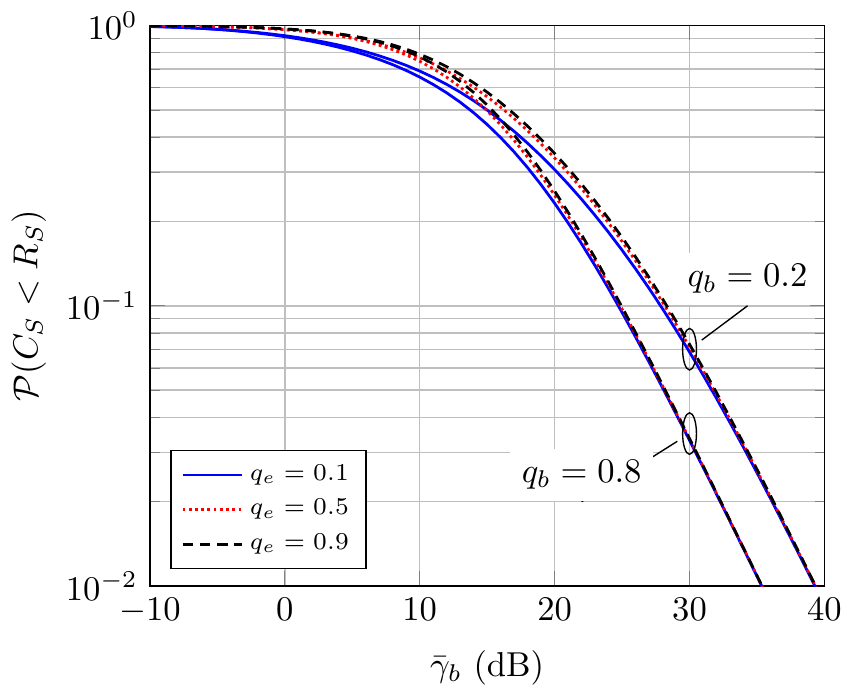}
\caption{Outage probability of secrecy capacity as a function of $\bar\gamma_b$, for different values of $q_e$ and $q_b$. Parameter values $\bar\gamma_e=15$ dB and $R_S=0.1$.}
\label{fig1}
\end{figure}

For a given value of $q_b$, we observe two different effects depending on the magnitude of $\bar\gamma_b$: in the low-medium SNR region, we see how a lower value of $q_e$ (i.e. a more severe fading in the eavesdropper link) makes the occurrence of a secrecy outage to be less likely. Hence, $\P(C_S<R_S)$ decreases with $q_e$ for a given $\bar\gamma_b$; we also note how the secrecy in this region is barely affected by the value of $q_b$.

Conversely, in the large SNR region we observe how the outage secrecy probability is mainly dominated by the fading severity of the desired link $q_b$. In this region, it is the distribution of $\gamma_b$ the dominant factor in the secure communication between Alice and Bob.

The probability of strictly positive secrecy capacity given in (\ref{eq:sec09}) is evaluated in Fig. \ref{fig2}, for the same set of parameter values considered in the previous figure.

\begin{figure}[lt]
\includegraphics[width=\columnwidth]{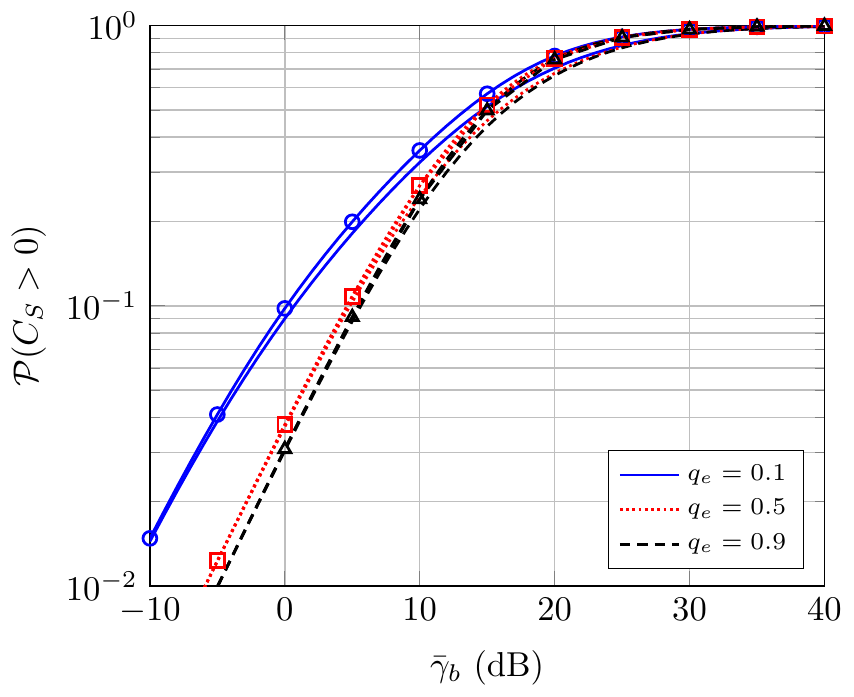}
\caption{Probability of strictly positive secrecy capacity as a function of $\bar\gamma_b$, for different values of $q_e$ and $q_b$. Parameter value $\bar\gamma_e=15$ dB. Solid lines only indicate $q_b=0.2$; solid lines with markers are included for $q_b=0.8$.}
\label{fig2}
\end{figure}

We can extract similar conclusions with regard of the effects of the fading severity in the desired and eavesdropper links. For low values of $\bar\gamma_b$, the secure communication is mainly determined by the distribution of $\gamma_e$; specifically, considering $\bar\gamma_b=5$ dB we see how $\P(C_S>0)$ is twice larger for $q_e=0.1$, compared to $q_e=0.5$. We also observe how a less severe fading in the desired link (i.e. a larger value of $q_b$) leads this probability to be larger.

\subsection{Outage Probability with interference}

Now, we consider the scenario analyzed in Section \ref{secOP}, where the OP of wireless links in Hoyt fading is investigated in the presence of arbitrarily distributed co-channel interference and background noise. We consider that the interference can be distributed according to the general $\eta$-$\mu$ and $\kappa$-$\mu$ distributions \cite{Yacoub2007}, widely employed in the literature for modeling NLOS and LOS propagation, respectively. For the sake of simplicity in the discussion, we consider a single interferer; however, note that the analysis introduced in Section \ref{secOP} can accommodate an arbitrary number of interferers, admitting correlation between them as well as combinations of $\eta$-$\mu$ and $\kappa$-$\mu$ distributions.

We will start considering that both co-channel interference and background noise are present, the desired link undergoes Hoyt fading, and the interference is distributed according to the $\eta$-$\mu$ or $\kappa$-$\mu$ distribution with arbitrary values of their parameters. We use format 2 for the $\eta$-$\mu$ distribution. We define the average SINR in this scenario as SINR$=\frac{\bar\gamma_d}{1+\bar\gamma_i}$, where $\bar\gamma_d$ is the ratio between the desired signal average receive power and the noise power (i.e., the SNR), whereas $\bar\gamma_i$ accounts for the ratio between the interferer average receive signal power at the receiver and the noise power (i.e., an interference-to-noise ratio, INR). 

Fig. \ref{fig3} evaluates the OP for different fading conditions for the desired and interfere links. We assume INR$=5$ dB, and a SINR threshold value used to declare an outage $\gamma_{o}=0$ dB. In the low SINR regime, it is possible to note the effect of the distribution of the interference in the OP: as $\mu$ is reduced, the propagation of the interfering signal experiences a more severe fade in both types of fading and hence the OP decays. Conversely, we observe that for large values of the average SINR the OP tends asymptotically to a value that is determined by the distribution of the desired link, i.e. by parameter $q$. For a given value of SINR, the OP grows as $q$ is reduced (i.e., an outage is declared with more probability when the direct link experiences a more severe fade).

\begin{figure}[lt]
\includegraphics[width=\columnwidth]{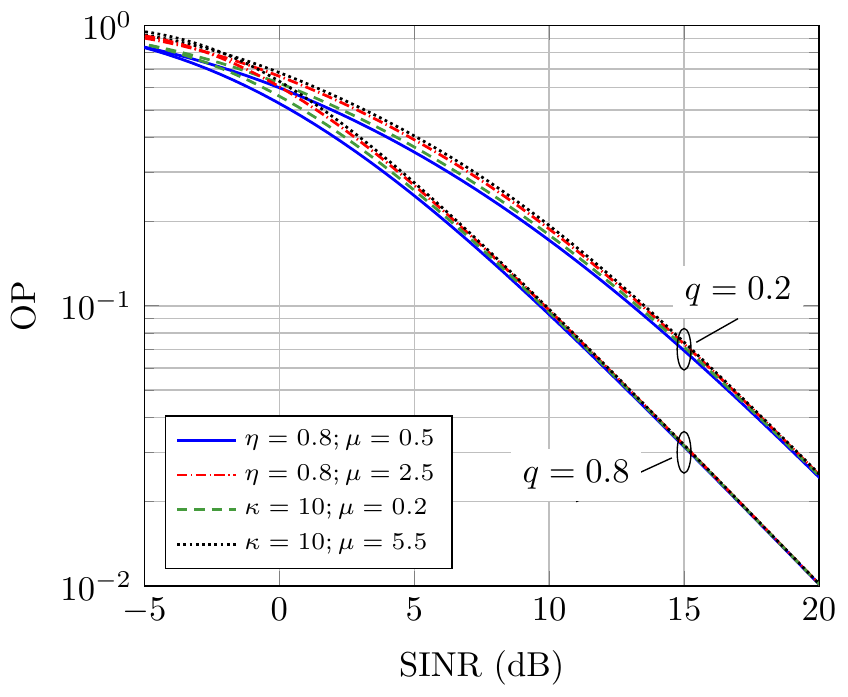}
\caption{Outage probability with co-channel interference and background noise as a function of the average SINR, considering a Hoyt distributed desired link and interference distributed according to $\eta$-$\mu$ or $\kappa$-$\mu$ distributions. Parameter values are INR$=5$ dB and $\gamma_{o}=0$ dB.}
\label{fig3}
\end{figure}

Now, we assume an interference-limited system, meaning that the background noise will be neglected. In this scenario, we define the average SIR as the ratio between the average powers of the desired and interfering signal. As in the previous case, we set the threshold value used to declare an outage as $\gamma_{o}=0$ dB. In Fig. \ref{fig4}, the OP is represented as a function of the SIR, for different values of parameters $q$, $\kappa$, $\eta$ and $\mu$. In the large SIR regime, it is reinforced the conclusion that the asymptotic OP is dominated by the distribution of the desired link. For low values of SIR, a similar behavior as in the previous figure is observed. However, it is noted a larger difference in the OP values when $\mu$ changes; this is due to the fact that when background noise vanishes, the differences between the different distributions of the interference become more evident.

\begin{figure}[lt]
\includegraphics[width=\columnwidth]{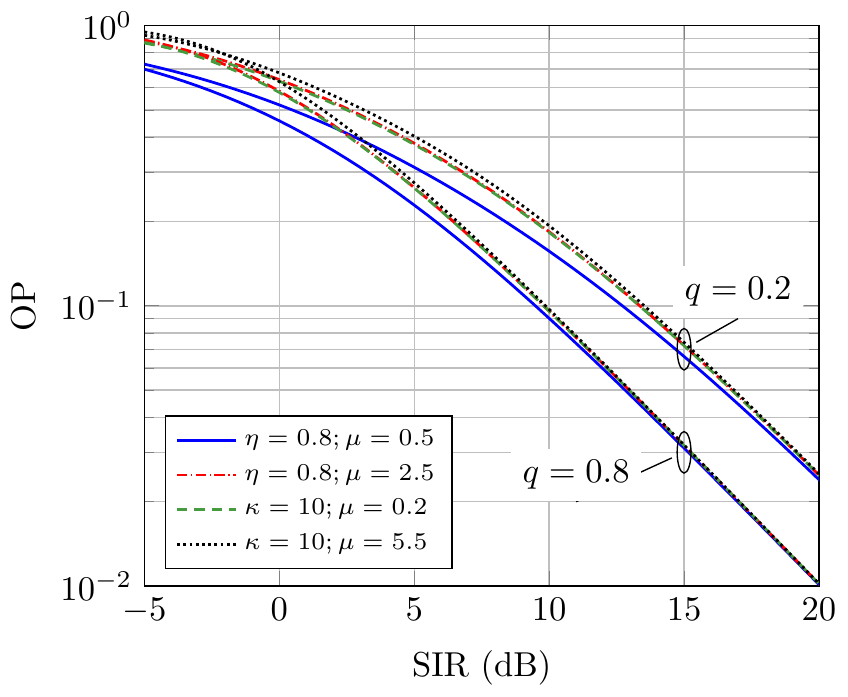}
\caption{Outage probability with co-channel interference as a function of the average SIR, considering a Hoyt distributed desired link and interference distributed according to $\eta$-$\mu$ or $\kappa$-$\mu$ distributions. Parameter value $\gamma_{o}=0$ dB.}
\label{fig4}
\end{figure}

%%%%%%%%%%%%%%%%%%%%%%%%%%%%%%%%%%%%%%%%%%%%%%%%%%%%%%%%%%%%%%%%%%%%%
%%%% 										SECTION VIII

\section{Conclusions}
\label{secConc}
We have provided a new look at the analysis of wireless communication systems in Hoyt (Nakagami-$q$) fading. Unlike previous approaches in the literature, we have found a connection between the Rayleigh and Hoyt distributions that facilitates the analysis in the latter scenario. 

By deriving integral expressions for the pdf and cdf of the squared Hoyt distribution, we have shown that the squared Hoyt distribution is in fact a composition of exponential distributions with continuously varying averages. Using this connection, we have introduced the Hoyt transform approach as a way to obtain easy-to-compute finite-range integral expressions of different performance metrics in Hoyt fading channels, as well as simple upper and lower bounds which become asymptotically tight as $q\rightarrow 1$, by simply leveraging existing results for Rayleigh fading channels. 

As a direct application, we have derived new expressions for several scenarios of interest in information and communication theory: (a) capacity analysis of adaptive transmission policies in Hoyt fading channels, (b) wireless information-theoretic security in Hoyt fading, and (c) outage probability analysis of Hoyt fading channels with \textit{arbitrarily distributed} interference and background noise.

%the outage probability in noise-limited environments as well as for scenarios with arbitrarily distributed co-channel interferers, which was an open problem in the literature. Simple new expressions for the channel capacity considering different transmission strategies have also been presented. We have shown that the presented framework permits to obtain average SER expressions for the cases when the MGF approach cannot be applied, and offers alternative expressions to those based on the MGF which are, in some cases, simpler. 

A further implication of the results in this paper is that there is no need to reproduce complicated calculations for Hoyt fading in the cases where tractable expressions are available for the Rayleigh case. Instead, these analyses can be easily extended to Hoyt scenarios by using a straightforward finite-range integral.

\section*{Acknowledgments}
The work of Juan M. Romero-Jerez was supported by the Spanish Government-FEDER public Project No. TEC2013-42711-R. The work of F.J. Lopez-Martinez was funded by Junta de Andalucia (P11-TIC-7109), Spanish Government-FEDER (TEC2013-44442-P, COFUND2013-40259), the University of Malaga and the European Union under Marie-Curie COFUND U-mobility program (ref. 246550).

\bibliographystyle{IEEEtran}
\bibliography{TIT}

%>\begin{IEEEbiographynophoto}{John Doe}
%>Biography text here.
%>\end{IEEEbiographynophoto}

%\newpage

%>\begin{IEEEbiographynophoto}{Jane Doe}

%>\end{IEEEbiographynophoto}

\end{document}